\documentclass{llncs}

\usepackage[applemac]{inputenc}
\usepackage[english]{babel}
\usepackage{latexsym}
\usepackage{amssymb}
\usepackage{amsmath}
\usepackage{ifthen,xspace}
\usepackage{tikz}
\usetikzlibrary{shapes, automata, positioning}
\usetikzlibrary[arrows]

\usepackage{enumerate}
\usepackage{theorem}
\usepackage{prettyref}

\newrefformat{thm}{Thm.~\ref{#1}}
\newrefformat{lem}{Lem.~\ref{#1}}
\newrefformat{dfn}{Def.~\ref{#1}}
\newrefformat{cor}{Cor.~\ref{#1}} 
\newrefformat{prop}{Prop.~\ref{#1}}
\newrefformat{sec}{Sec.~\ref{#1}}
\newrefformat{ex}{Ex.~\ref{#1}}
\newrefformat{fig}{Fig.~\ref{#1}}
\newcommand{\prref}[1]{\prettyref{#1}}

\newrefformat{alg}{Algorithm~\ref{#1}}
\newrefformat{thm}{Theorem~\ref{#1}}
\newrefformat{lem}{Lemma~\ref{#1}}
\newrefformat{def}{Definition~\ref{#1}}
\newrefformat{cor}{Corollary~\ref{#1}}
\newrefformat{prop}{Proposition~\ref{#1}}
\newrefformat{sec}{Section~\ref{#1}}
\newrefformat{subsec}{Subsection~\ref{#1}}
\newrefformat{kap}{Chapter~\ref{#1}}
\newrefformat{ex}{Example~\ref{#1}}
\newrefformat{rem}{Remark~\ref{#1}}
\newrefformat{fig}{Figure~\ref{#1}}
\newrefformat{app}{Appendix~\ref{#1}}
\newrefformat{eq}{Equation~(\ref{#1})}
\newrefformat{tab}{Table~\ref{#1}}
\newrefformat{ap}{{\sc Appendix}}

\newcommand{\LTL}{\ensuremath{\mathrm{LTL}}}
\newcommand{\LTLS}{\ensuremath{\LTL_\Sigma}}

\newcommand{\PSPACE}{\ensuremath{\mathrm{PSPACE}}\xspace}

\newcommand{\True}{\top}

\newcommand{\XU}{\mathbin{\mathsf{XU}}}

\newcommand{\compl}{\mathfrak {c0}}
\newcommand{\set}[2]{\{#1 \mid #2\}}

\newcommand{\os}[1]{\{\mathinner{#1}\}}

\newcommand{\abs}[1]{\left|\mathinner{#1}\right|}

\newcommand{\N}{\mathbb{N}}

\newcommand{\Q}{\mathbb{Q}}
\newcommand{\R}{\mathbb{R}}

\newcommand{\cA}{\mathcal{A}}
\newcommand{\cB}{\mathcal{B}}

\newcommand{\cR}{\mathcal{R}}

\newcommand{\cM}{\mathcal{M}}

\newcommand{\cP}{\mathcal{P}}

\newcommand{\cO}{\mathcal{O}}

\newcommand{ \ov}[1]{ \overline{#1}\, }
\newcommand{\IFF}{if and only if\xspace}

\newcommand{\epi}{epimorphism\xspace}

\newcommand{\sse}{\subseteq}
\newcommand{\es}{\emptyset}
\newcommand{\sm}{\setminus}

\renewcommand{\phi}{\varphi}

\newcommand{\oo}{\omega}
\newcommand{\alp}{\alpha}
\newcommand{\bet}{\beta}
\newcommand{\gam}{\gamma}
\newcommand{\del}{\delta}
\newcommand{\sig}{\sigma}

\newcommand{\Sig}{\Sigma}
\newcommand{\Gam}{\Gamma}
\newcommand{\Del}{\Delta}

\newcommand{\Aa}{\mathcal{A}}
\newcommand{\Bb}{\mathcal{B}}

\newcommand{\act}[1]{\stackrel{#1}{\longrightarrow}}

\newcommand{\bean}{\begin{eqnarray*}}
\newcommand{\eean}{\end{eqnarray*}}

\newcommand{\moni}{monitorable\xspace}

\newcommand{\bin}{\mathop{\mathrm{bin}}}

\begin{document}
\title{A Note on Monitors and B\"uchi automata}

\author{Volker Diekert\inst{1} and Anca Muscholl\inst{2} and Igor Walukiewicz\inst{2}}
\institute{Universit\"at Stuttgart, FMI, Germany 
 \and LaBRI, University of Bordeaux, France}

\maketitle
\begin{abstract}
 When a property needs to be checked against an unknown or very
 complex system, classical exploration techniques like model-checking
 are not applicable anymore. Sometimes a~monitor can be
 used, that checks a given property on the underlying system at runtime. A
 monitor for a property $L$ is a deterministic finite automaton $\cM_L$ that after each
 finite execution tells whether (1) every possible extension of the
 execution is in $L$, or (2) every possible extension is in the
 complement of $L$, or neither (1) nor (2) holds. Moreover, $L$ being
 monitorable  means that it is  always possible that in some future the monitor reaches (1) or (2). Classical examples
 for monitorable properties are safety and cosafety properties. On the other hand,  deterministic liveness properties like 
 ``infinitely many $a$'s'' are not monitorable.

 We discuss various monitor constructions with a focus on deterministic $\oo$-regular languages. We locate a proper subclass of 
 of deterministic $\oo$-regular languages but also strictly large than the 
 subclass of languages which are deterministic and codeterministic; and for this subclass there exists a canonical monitor which also accepts the language itself. 

 We also address the problem to decide 
 monitorability in comparison with deciding liveness. 
  The state of the art is as follows. Given a B\"uchi automaton, it is PSPACE-complete to decide liveness or
 monitorability. Given an LTL formula, deciding liveness becomes 
 EXPSPACE-complete, but  the complexity to decide monitorability remains open.
\end{abstract}

\section*{Introduction}\label{sec:intro}
Automata theoretic verification has its mathematical 
foundation in classical papers written in the 1950's and 1960's by B\"uchi, Rabin and others. Over the past few decades it became a success story with large scale industrial applications. However, frequently properties need to be checked against an unknown or very
complex system. In such a situation 
classical exploration techniques like model-checking might fail.  The model-checking problem asks whether all runs satisfy a given
specification. If the specification is written in  monadic second-order logic, then all runs obeying the specification 
can be expressed effectively by some B{\"u}chi automaton (BA for short). 
If the abstract model of the system is given by some finite transition system, then the model-checking problem becomes an inclusion
problem on $\omega$-regular languages: all runs of the transition system 
must be accepted by the BA for the specification, too. 
In formal terms we wish to check $L(\cA) \sse L(\phi)$ where 
$\cA$ is the transition system of the system and $\phi$ is a formula for the specification. Typically testing inclusion is expensive, hence 
it might be better to check the equivalent assertion $L(\cA) \cap L(\neg \phi) = \es$.
 This is
a key fact, because then the verification problem becomes a 
reachability problem in finite graphs.

Whereas the formulas are typically rather small, so we might be able to construct the B\"uchi automaton for $L(\neg \phi)$, the transition systems tend to be very large.
Thus, ``state explosion'' on the system side might force us to use weaker concepts.
The idea is to construct a ``monitor'' for a given 
specification. A monitor observes the system during runtime. 
It is a finite deterministic automaton with at most two distinguished states 
$\bot$ and $\top$. If it reaches the state $\bot$, the monitor stops 
and raises an ``alarm'' that no continuation of the so far observed run will satisfy the specification. If it reaches $\top$, the monitor stops because all continuations will satisfy the specification. Usually, this means we must switch to a finer monitor. Finally, we say that a language is monitorable, if 
in every state of the monitor it is possible to reach either $\bot$ or $\top$
or both. 

The formal definition 
of \moni properties has been given in \cite{PnueliZ06} by
Pnueli and Zaks. It generalizes the notion of a \emph{safety property}
because for a safety property some deterministic finite automaton can raise an alarm  $\bot$ by observing a finite ``bad prefix'', 
once the property is violated. 
The extension to the more general notion of monitorability is that a \moni property gives also a positive 
feedback $\top$, if all extensions of a finite prefix obey the specification. 
 Monitors are sometimes easy 
to implement and have a wide range of applications. See for example \cite{LeuckerS08jlap} and the references therein. 
Extensions and applications for stochastic automata have been proposed in 
Sistla et al., see \cite{GondiPS09,SistlaZF11}.

In the present paper we discuss various monitor constructions. 
%We emphasize that in many cases very small monitors exist, but then they become essentially useless. So, the challenge is to define a ``good'' monitor. 
A monitor for a safety property $L$ can have much less states than the smallest 
DBA accepting $L$. For example, let $\Sig = \os{a,b}$ and $n \in \N$.
Consider the language $L= a^nba\Sig^\oo \sm \Sig^* bb \Sig^\oo$. The reader is invited to check that $L$ is a safety property and every DBA accepting $L$ has more than $n$ states. But there is monitor with three states, only. The monitor patiently waits to see an occurrence of a factor  $bb$ and then switches to 
$\bot$. Hence, there is no bound between a 
minimal size of an accepting DBA and the minimal size of a possible monitor. This option has been actually one of the main motivations to introduce the notion of monitor. 

There are many deterministic languages which are far away from being \moni. 
Consider again $\Sig = \os{a,b}$ and let $L$ be the deterministic language 
of ``infinitely many $a$'s''. It is shown in \cite{DiekertL2014tcs} that $L$ cannot be written as any countable union of \moni languages. 
On the other hand, if $L$ is \moni and also accepted by some DBA with $n$ states and a single initial state, then there is monitor accepting $L$ with at most $n$ states. 

In the last section of this paper we discuss the question how to decide whether a language is \moni
and its complexity. If the input is a B\"uchi automaton, then deciding safety, liveness, or monitorability is PSPACE-complete.
If the input is an $\LTL$ formula, then deciding safety remains 
PSPACE-complete. It becomes surprisingly difficult for liveness: 
EXPSPACE-complete. For monitorability the complexity is wide open: we only know that is PSPACE-hard and that it can be solved in EXPSPACE.

%%%%%%%%%%%%%%%%%%%%%%%%%%%%%%%%%%%%%%
\section{Preliminaries}\label{sec:prem}
We assume that the reader is familiar with the basic facts about 
automata theory for infinite words as it is exposed in the survey 
\cite{tho90handbook}.  In our paper $\Sig$ denotes a finite nonempty alphabet. 
We let $\Sig^*$ (resp.~$\Sig^\oo$) be the set of finite (resp.~infinite) words over $\Sig$. Usually, lower case letters like $a$, $b$, $c$ denote letters 
in $\Sig$, $u, \ldots, z$ denote finite words, $1$ is the empty word, and
 $\alp$, $\bet$, $\gam$ denote infinite words. By  language we mean a subset $L \sse \Sig^\oo$. The complement of $L$ w.r.t.~$\Sig^\oo$ is denoted by 
 $L^{\compl}$. Thus, $L^{\compl} = \Sig^\oo \sm L$. 

A \emph{B{\"u}chi automaton} (\emph{BA} for short) is a tuple $\cA= (Q,\Sig,\del,I,F)$
where $Q$ is the nonempty finite set of states, $I\sse Q$ is the set of initial states,
$F\sse Q$ is the set of final states, and $\del \sse Q\times \Sig \times Q$ is the transition relation. 
The accepted language $L(\cA)$ is the set of infinite words $\alp \in \Sig^\oo$ which label an infinite path in $\cA$ which begins at some state in $I$ and visits some state in $F$ infinitely often. Languages of type $L(\cA)$ are called \emph{$\oo$-regular}.

If for each $p\in Q$ and $a \in \Sig$ there is at most one $q\in Q$ with $(p,a,q)\in \del$, then $\cA$ is called 
 \emph{deterministic}. We write \emph{DBA} for deterministic B{\"u}chi automaton. In a DBA we view $\del$ as a partially defined function and we also write $p\cdot a= q$ instead of  $(p,a,q)\in \del$. Frequently it is asked that 
 a DBA has a unique initial state. This is not essential, but in order to follow the standard notation $(Q,\Sig,\del,q_0,F)$ refers to a BA where $I$ is the singleton $\os{q_0}$.

A \emph{deterministic weak B{\"u}chi automaton} (\emph{DWA} for
short) is a DBA  where all states in a strongly connected
component are either final or not final. Note that a strongly connected component may have a single state because the underlying directed graph may have self-loops. 
A language is accepted by some DWA \IFF it is deterministic and simultaneously codeterministic. The result is in \cite{Staiger83} which in turn is based on previous  papers by Staiger and Wagner \cite{sw74eik} and Wagner \cite{Wagner79}.

According to \cite{PnueliZ06} a \emph{monitor}
is a finite deterministic transition system $\cM$
with at most two 
distinguished states $\bot$ and $\top$ such that for all states $p$
either there exist a path from $p$ to $\bot$, or to $\top$, or to both. 
It is a \emph{monitor for an $\oo$-language $L \sse \Sig^\oo$} if 
the following additional properties are satisfied:
\begin{itemize}
\item If $u$ denotes the label of a path from an initial state to $\bot$, then 
$u\Sig^\oo \cap L = \es$.
\item If $u$ denotes the label of a path from an initial state to $\top$, then 
$u\Sig^\oo \sse L$.
\end{itemize}
A language $L \sse \Sig^\oo$ is called \emph{monitorable} if there exists a
monitor for $L$. 
Thus, even non regular languages might be \moni. 
If a property is monitorable, then the following holds:
\begin{equation}\label{eq:neccon}
\forall x\, \exists w: xw\Sigma^\oo \subseteq L \vee xw\Sigma^\oo
\cap L=\es\,.
\end{equation}
The condition in (\ref{eq:neccon}) is not sufficient for non-regular languages: indeed
consider $L = \set{a^nb^na}{n\in \N}\Sig^\oo$. There is no finite state monitor for this language.
In the present paper, the focus is on monitorable $\oo$-regular languages.
For $\oo$-regular languages (\ref{eq:neccon}) is also sufficient; and \prref{rem:staig} below shows an equivalent condition for
monitorability (although stronger for non-regular languages).

The common theme in ``automata on infinite words'' is that finite state devices serve to classify $\omega$-regular properties.
The most prominent classes are:
\begin{itemize}
\item \emph{Deterministic properties:} 
there exists a DBA.
\item \emph{Deterministic properties which are simultaneously codeterministic:}  
there \goodbreak exists a DWA.  
\item \emph{Safety properties:} there exists a DBA where all states are final. 
\item \emph{Cosafety properties:} the complement is a safety property.
\item \emph{Liveness properties:} there exists a BA where from all states 
there is a path to some final state lying in a strongly connected component.
\item \emph{Monitorable properties:} there exists a monitor. 
\end{itemize}
According to our definition of a monitor, not both states $\bot$ and $\top$ need to be defined. Sometimes it is enough to see $\bot$ or $\top$. For example, 
let $\es \neq L\neq \Sig^\oo$ be a safety property and $\cA= (Q, \Sig, \del, I,Q)$ be a DBA accepting $L$ where all states are final. Since 
$\es \neq L$ we have $I \neq \es$. Since $L\neq \Sig^\oo$, the partially defined transition function $\del$ is not defined everywhere. Adding a state
$\bot$ as explained above turns $\cA$ into a monitor $\cM$ for $L$ where the state
space is $Q \cup \os \bot$. There is no need for any state $\top$.
The monitor $\cM$ also accepts $L$. This is however not the general case.
%%%%%%%%%%%%%%%%%%%%%%%%%%%%%%%%%%%%%%
\section{Topological properties}\label{sec:tp}
A topological space is a pair $(X,\cO)$ where $X$ is a set and
$\cO$ is collection of subsets of $X$ which is closed under arbitrary 
unions and finite intersections. In particular, $\es, X \in \cO$. A subset $L \in \cO$ is called \emph{open}; and its complement $X\sm L$  is called \emph{closed}.

For $L \sse X$ we denote by $\ov L$ the intersection over all 
closed subsets $K$ such that $L \sse K \sse X$. 
It is the \emph{closure} of $L$. 
The complement $X\sm L$ is denoted by 
$L^{\compl}$. 

A subset $L\sse X $ is called \emph{nowhere dense} if its closure $\ov L$ does not contain any open subset. The classical example of the uncountable Cantor set $C$ inside the closed interval 
$[0,1]$ is nowhere dense. It is closed and does not have any open subset.
On the other hand, 
the subset of rationals $\Q$ inside $\R$ (with the usual topology) satisfies $\ov \Q = \R$. Hence, $\Q$ is ``dense everywhere'' although $\Q$ itself 
does not have any open subset.

The \emph{boundary} of $L$ is sometimes denoted as $\del(L)$; it is defined 
by $$\del(L)= \ov L \cap \ov{L^{\compl}}.$$

In a metric space $B(x,1/n)$ denotes the \emph{ball of radius $1/n$}.
It is the set of $y$ where the distance between $x$ and $y$ is less than $1/n$.  A set is open \IFF it is some union of balls, and the closure of $L$ can be written as   
$$\ov L = \bigcap_{n\geq 1}\; \bigcup_{x \in L}B(x,1/n).$$
In particular, every closed set is a countable intersection of open sets. 
Following the traditional notation we let $F$ be the family 
of closed subsets and $G$ be 
the family of open subsets. Then  $F_{\sig}$ denotes the family 
of countable unions of closed subsets  and 
$G_{\del}$ denotes the family 
of countable intersections of open subsets. We have just seen $F \sse G_\del$,  and we obtain $G \sse F_\sig$ by duality. Since 
$G_{\del}$ is closed under finite union, $G_\del\cap F_\sig$ is Boolean algebra which contains all open and all closed sets. 

In this paper we deal mainly with $\oo$-regular sets. These are subsets 
of $\Sig^\oo$; and $\Sig^\oo$ is endowed with a natural topology 
where the open sets are defined by the sets of the form $W\Sig^\oo$ where $W\sse \Sig^*$. It is called the \emph{Cantor topology}. The Cantor topology
corresponds to a complete ultra metric space: for example, we let $d(\alp,\bet) =1/n$ for $\alp, \bet \in \Sig^\oo$ where $n-1\in \N$ is the length of a maximal common prefix of $\alp$ and $\bet$. (The convention is $0= 1 /\infty$.)

The following dictionary translates notation about $\oo$-regular sets
into its topological counterpart. 
\begin{itemize}
\item Safety = closed sets = $F$.  
\item Cosafety = open sets = $G$.
\item Liveness = \emph{dense} = closure is $\Sig^\oo$. 
\item Deterministic =  $G_{\del}$, see \cite{Landweber69}.
\item Codeterministic =  $F_{\sig}$, by definition and the previous line. 
\item Deterministic  and simultaneously codeterministic = 
$G_\del\cap F_\sig$, by definition. %\cite{Staiger83}
\item Monitorable = the boundary is nowhere dense, see  \cite{DiekertL2014tcs}.
\end{itemize}

Monitorability depends on the ambient space $X$. Imagine we embed
$\R$ into the plane $\R^{2}$ in a standard way. Then $\R$ is a line which is nowhere dense in $\R^{2}$. As a consequence every subset $L\sse \R$ is \moni in $\R^{2}$. The same phenomenon happens  for $\oo$-regular languages.
Consider the embedding of $\os{a,b}^\oo$ into $\os{a,b,c}^\oo$ by choosing 
a  third letter $c$. Then $\os{a,b}^\oo$ is nowhere dense in $\os{a,b,c}^\oo$ and hence, every subset $L \sse \os{a,b}^\oo$ is \moni in $\os{a,b,c}^\oo$. The monitor has $3$ states. One state is initial and by reading $c$ 
we switch into the state $\bot$. The state $\top$ can never be reached. 
In some sense this $3$-state minimalistic monitor is useless: it tells us almost nothing 
about the language. Therefore the smallest possible monitor is rarely the best one.  

\begin{remark}\label{rem:forfact}
In our setting many languages are \moni because there exists a ``forbidden factor'', for example a letter $c$ in the alphabet which is never used. More precisely, let $L \sse \Sig^\oo$ be any subset and assume that there exists a finite word $f \in \Sig^*$
such that either $\Sig ^* f \Sig^\oo \sse L$ or $\Sig ^* f \Sig^\oo \cap L= \es$. Then $L$ is \moni. 
Indeed, the monitor just tries to recognize $\Sig ^* f \Sig^\oo$. Its size is $\abs f +2$ and can be constructed in linear time from $f$ by algorithms of Matiyasevich \cite{mat73} or Knuth-Morris-Pratt \cite{KMP}.
\end{remark}

\section{Constructions of monitors}\label{sec:cm}
\prref{rem:forfact} emphasizes that one should not try simply to minimize monitors. The challenge is to construct ``useful'' monitors. In the extreme, think that we encode a language $L$ in printable ASCII code, hence it is a subset of $\os{0,1}^*$. 
But even in using a $7$-bit encoding there were $33$ non-printable characters. 
A monitor can choose any of them and then waits patiently whether this
very special encoding error ever happens. This might be a small monitor,
but it is of little interest. It does not even check all basic syntax
errors. 

\subsection{Monitors for $\oo$-regular languages in $G_\del \cap F_\sig$}\label{sec:stcm}
The $\oo$-regular languages in $G_\del \cap F_\sig$ are those which are deterministic and simultaneously codeterministic. In every complete metric space (as for example the Cantor space $\Sig^\oo$) 
all sets in 
$G_\del\cap F_\sig$ have a boundary which is nowhere dense.  Thus, deterministic  and simultaneously codeterministic languages are \moni by a purely topological 
observation, see \cite{DiekertL2014tcs}.

Recall that there is another characterization of $\oo$-regular languages in $G_\del \cap F_\sig$ due to Staiger, \cite{Staiger83}. 
It says that these are the languages which are accepted by some DWA, thus by some DBA where in every strongly connected component either all states are final or none is final.

In every finite directed graph  there is at least one  strongly connected component which cannot be left anymore. In the minimal DWA (which exists and which is unique and where, without restriction, the transition function is totally defined)
these end-components consist of a single state which can be identified either with $\bot$ or with $\top$. Thus, the DWA is itself a monitor. 
Here we face the problem that this DWA might be very large and also too complicated for useful monitoring.

\subsection{General constructions}\label{sec:gcm}
Let $w \in \Sig^*$ be any word. Then the language  $L= w \Sig^\oo$ is \emph{clopen} meaning simultaneously open and closed. The minimal monitor for $w \Sig^\oo$ must read the whole word $w$ before it can make a decision; and the minimal monitor has exactly $\abs w +2$ states. 
On the other hand, its boundary, $\ov L \cap \ov{L^{\compl}}$ is empty and therefore nowhere dense. 
This suggests that deciding monitorability might be much simpler than constructing a monitor. 
For deciding we just need any DBA accepting the safety property $\ov L \cap \ov{L^{\compl}}$. Then we can see on that particular DBA whether $L$ is \moni, although this particular DBA might be of no help for monitoring. 
Phrased differently, there is no bound between the size of a DBA certifying that $L$ is \moni and the size of an actual monitor for $L$. 
 
Indeed, the standard construction for a monitor $M_{L}$ is quite different {}from a direct construction of the  DBA for 
the boundary, see for example \cite{DiekertL2014tcs}. 
The construction for the monitor $M_{L}$ is as follows. Let $L\sse \Sig^\oo$ be \moni and given by some BA. First, we construct
two DBAs: one DBA with state set $Q_1$, for the closure $\ov L$ and another one with state set $Q_2$ for the closure of the complement $\ov{L^{\compl}}$. 
We may assume that in both DBAs all states  are final and reachable from a unique initial state $q_{01}$ and $q_{02}$, respectively. Second,  let $Q'= Q_1 \times Q_2$. Now, if we are in a
state $(p,q)\in Q'$ and we want to read a letter $a\in \Sig$, then exactly one out of the three
possibilities can happen. 
\begin{enumerate}
\item The states $p\cdot a$ and  $q\cdot a$ are defined, in which case we let $(p,q)\cdot a = (p\cdot a,q\cdot a)$.
\item The state $p\cdot a$ is not defined, in which case we let $(p,q)\cdot a =\bot$.
\item The state $q\cdot a$ is not defined, in which case we let $(p,q)\cdot a = \top$.
\end{enumerate}
Here $\bot$ and $\top$ are new states. Moreover, we let $q\cdot a= q$ for $q \in \os{\bot,\top}$ and $a \in \Sig$. 
Hence, the transition function is totally defined.
Finally, we let $Q\sse Q'\cup \os{\bot,\top}$ be the subset which is reachable from the initial state $(q_{01},q_{02})$.
Since $L$ is \moni, 
$Q \cap \os{\bot,\top}\neq \es$; and $Q$ defines a set of a monitor $\cM_L$.
Henceforth, the monitor $\cM_L$ above is called a \emph{standard monitor for $L$}. 
The monitor has exactly one initial state. {}From now on, 
for simplicity, we assume that every monitor $\cM$ has exactly one initial state and that the transition function is totally defined. Thus, we can denote a monitor $\cM$ as a tuple
\begin{equation}\label{eq:initmoni}
\cM = (Q,\Sig,\del,q_0,\bot,\top).
\end{equation}
Here, $\del:Q \times \Sig \to Q,\, (p,a) \mapsto p \cdot a$ is the  transition function, 
$q_0$ is the unique initial state, $\bot$ and $\top$ are distinguished states with $Q \cap \os{\bot,\top} \neq \es$.

\begin{definition}\label{def:morphmoni}
Let $\cM = (Q,\Sig,\del,q_0,\bot,\top)$, $\cM' = (Q',\Sig,\del',q'_0,\bot,\top)$ be monitors. A \emph{morphism} between $\cM$ and $\cM'$ 
is mapping $\phi: Q \cup \os{\bot,\top} \to Q' \cup \os{\bot,\top}$ such that
$\phi(q_0) = q'_0$, $\phi(\bot) = \bot$, $\phi(\top) = \top$, and
$\phi(p\cdot a) =\phi(p)\cdot a$ for all $p\in Q$ and $a \in \Sig$. 

If $\phi$ is surjective, then $\phi$ is called an \emph{epimorphism}.
\end{definition}

Another canonical monitor construction uses the classical notion of right-congruence. 
A \emph{right-congruence} for the monoid $\Sig^*$ is an equivalence relation 
$\sim$ such that $x\sim y$ implies  $xz\sim yz$ for all $x,y,z \in \Sig^*$.
There is a canonical right-congruence $\sim_{L}$ associated with every
$\oo$-language $L \sse \Sig^\oo$:
for $x\in \Sig^*$ denote by $L(x) = \set{\alp\in \Sig^\oo}{x\alp \in L}$ the  \emph{quotient} of $L$ by $x$. 
 Then defining $\sim_{L}$ by $x\sim_{L} y \iff L(x) = L(y)$ yields a right-congruence. 
 More precisely, 
 $\Sig^*$ acts on 
 the set of quotients $Q_L = \set{L(x)}{x\in \Sig^*}$ on the right, and the formula for the action becomes $L(x)\cdot z =  L(xz)$. Note that this is well-defined.
 This yields the \emph{associated automaton} \cite[Section 2]{Staiger83}. It the finite deterministic transition system with 
 state set $Q_{L}$ and arcs $(L(x),a,L(xa))$ where 
 $x \in \Sig^*$ and $a \in \Sig$.

 There is a canonical initial state $L= L(1)$,
 but unlike in the case of regular sets over finite words there is no good notion of final states in $Q_L$ for infinite words. 
 The right congruence is far too coarse to recognize $L$, in general. 
For example, consider the deterministic language $L$ of ``infinitely many $a$'s'' in $\os{a,b}^\oo$. For all $x$ we have $L=L(x)$, but in order to recognize $L$ we need two states.

It is classical that if $L$ is $\oo$-regular, then the set $Q_L$ is finite, but the converse fails badly \cite[Section 2]{Staiger83}: 
 there are uncountably many languages where 
 $\abs{Q_L} = 1$. To see this define for each $\alp \in \Sig^\oo$ a set
 $$L_{\alp}= \set{\bet\in \Sig^\oo}{\alp \text { and} \bet \text{ share an infinite suffix}}.$$ All $L_{\alp}$ are countable, but the union 
$\set{L_{\alp}}{\alp\in \Sig^\oo}$ 
  covers the uncountable Cantor space $\Sig^\oo$. Hence, there are uncountably many $L_{\alp}$. However, $\abs{Q_{L_\alp}} = 1$ since $L_{\alp}(x) =L_{\alp}$ for all $x$. 
 
Recall that a monitor is a DBA where the monitoring property is not defined 
using final states, but it is defined using the states $\bot$ and $\top$. 
Thus, a DBA with an empty set of final states can be used as a monitor as long as $\bot$ and $\top$ have been assigned and the required properties for a monitor are satisfied.

\begin{proposition}\label{prop:moniquot}
Let $L \sse \Sig^\oo$ be $\oo$-regular and monitorable. 
Assume that $L$ is accepted by some BA with $n$ states. 
As above let $Q_L = \set{L(x)}{x\in \Sig^*}$ and denote $\bot = \es$ and 
$\top = \Sig^\oo$. Then $\abs{Q_L} \leq 2^n$ and
$Q_{L} \cup \os{\top,\bot}$ is the set of states for a monitor for $L$. At least one of the states in $\os{\top,\bot}$ is reachable from the initial state $L= L(1)$. 
\end{proposition}
%
%\begin{proof}
% Trivial.
%\end{proof}
The monitor in \prref{prop:moniquot} with state space $Q_L$ is denoted by $\cA_L$ henceforth. We say that $\cA_L$ is the \emph{right-congruential} monitor for $L$.

\begin{proposition}\label{prop:epi}
Let $\cA$ be the right-congruential monitor for $L$. 
Then the mapping $$L(x) \mapsto \phi(L(x))= (\ov L(x),\ov {L^{\compl}}(x))$$
induces a canonical \epi from $\cA_L$ onto some standard monitor $\cM_L$.\end{proposition}

\begin{proof}
Observe that $\ov L(x) = \ov{L(x)}$ and 
$L^{\compl}(x) = L(x)^\compl$. Hence, $(\ov L(x),\ov {L^{\compl}}(x)) = (\ov {L(x)},\ov {L(x)^\compl})$ and $\phi(L(x))$ is well-defined. Now, if
$\ov L(x)\neq \es$ and $\ov {L^{\compl}}(x) \neq \es$, then $\phi(L(x))\in Q$ where 
$Q$ is the state space of the standard monitor $\cM$. 
If $\ov L(x)=\es$ then we can think that all $(\es,\ov {L^{\compl}}(x))$ denote the 
state $\bot$; and if $\ov {L^{\compl}}(x)=\es$ then we can think that all $(\ov L(x),\es)$ denote the 
state $\top$.
\qed
\end{proof}
\begin{corollary}\label{cor:moniexpstan}
Let $L \sse \Sig^\oo$ be monitorable and given by some BA with $n$ states. 
Then some standard  monitor $\cM_{L}$ for $L$ has at most $2^n$
states. 
\end{corollary}

\begin{proof}
Without restriction we may assume that in the BA $(Q,\Sig, \del,I,F)$ accepting $L$ every state $q\in Q$ leads to some final state. The usual subset construction leads first  to a DBA accepting  $\ov {L}$,
where all states are final and the states of this DBA are the nonempty subsets 
of $Q$. Thus, these are $2^n-1$ states. 
Adding the empty set $\es = \bot$ 
we obtain a DBA with $2^n$ states where the transition function is defined everywhere. If the complement $L^{\compl}$ is dense, this yields a standard monitor. In the other case we can use the subset construction also for a DBA 
accepting $\ov {L^{\compl}}$. 
In this case we remove all subsets $P \sse Q$ 
where $L(Q,\Sig, \del,P,F) = \Sig^\oo$. (Note, for all $a \in \Sig$ we have: 
if $L(Q,\Sig, \del,P,F) = \Sig^\oo$ and $P'= \set{q\in Q}{\exists p \in P:(p,a,q)\in \del}$, then $L(Q,\Sig, \del,P',F) = \Sig^\oo$, too.) 
Thus, if $L^{\compl}$ is not dense, then  the construction for a standard monitor has 
at most $2^n-2$ states of the form $(P,P)$ where $\es \neq P$ and $L(Q,\Sig, \del,P,F) \neq \Sig^\oo$. In addition there exists the reachable state 
$\top$ and possibly the state $\bot$. 
\qed\end{proof}

\prref{prop:epi} leads to the question of a canonical minimal monitor, at least 
for a safety language where a minimal accepting DBA exists. The answer
is ``no'' as we will see in \prref{ex:anb} later. 

Let us finish the section with a result on arbitrary \moni subsets of $\Sig^\oo$
which is closely related to  \cite[Lemma 2]{Staiger76a}. 
Consider any subset $L \sse \Sig^\oo$ where the set of quotients
$Q_L = \set{L(x)}{x\in \Sig^*}$ is finite (=``zustandsendlich'' or ``finite state''in the terminology of \cite{Staiger76a}). If $Q_L$ is finite, then  
$L$ is \moni \IFF the 
boundary is nowhere dense. In every topological space this latter condition is equivalent to the condition that 
the interior of $L$ 
is dense in its closure $\ov L$. Translating Staiger's result in 
\cite{Staiger76a} to the notion of monitorability we obtain the 
following fact.  
\begin{proposition}%[c.f.~\cite{Staiger76a}]
\label{prop:staiger}
Let $L \sse \Sig^\oo$ be any monitorable language and let $\cM$ be a monitor for $L$ with $n$ states. 
Then there exists a finite word $w$ of length at most $(n-1)^2$ such that
for all $x\in \Sig^*$ we have either $xw\Sig^\oo \sse L$  or 
$xw\Sig^\oo \cap L = \es$.
\end{proposition}

\begin{proof} We may assume that $n \geq 1$ and that the state space of $\cM$ is included in $\os{1,\ldots, n-1, \bot, \top}$. Merging $\top$ and $\bot$ into a single state 
$0$ we claim that there is a word $w$ of length at most $(n-1)^2$
such that $q\cdot w = 0$ for all $0\leq q \leq n-1$. 
Since $L$ is \moni, there is for each $q\in \os{0,\ldots, n-1}$ a finite word $v_q$ of length at most $n-1$ such that $q\cdot v_q = 0$.
By induction on $k$ we may assume that there is a word 
$w_k$ of length at most $k(n-1)$ such that for each 
$q\in \os{0,\ldots, k}$ we have $q\cdot w_k =0$. (Note that the assertion trivially holds for $k=0$.) If $k \geq  n-1$ we are done: $w = w_{n-1}$. Otherwise 
consider the state $q= k+1$ and the state $p = q\cdot w_k$. 
Define the word $w_{k+1}$ by $w_{k+1} = w_k v_p$. Then the length of $w_{k+1}$ is at most $(k+1)(n-1)$. Since $w_{k}$ is a prefix of $w_{k+1}$ and since $0\cdot v = v$ for all $v$, we have
$q\cdot  w_{k+1}= 0$ for all $0\leq q \leq k+1$. \qed
\end{proof}

\begin{remark}\label{rem:staig}
The interest in \prref{prop:staiger} is that monitorability can be characterized by a single alternation of quantifiers.  
Instead of saying that
$$\forall x\, \exists w\, (\forall \alp : xw \alp \in L) \vee (\forall \alp : xw \alp \notin L)$$
it is enough to say 
$$ \exists w\, \forall x\,(\forall \alp : xw \alp \in L) \vee (\forall \alp : xw \alp \notin L).$$
The length bound $(n-1)^2$ is not surprising. It confirms {\v C}ern{\'y}'s Conjecture in the case of monitors. (See \cite{Volkov08} for a survey on {\v C}ern{\'y}'s Conjecture.) Actually, in the case of monitors with more than $3$ states the estimation of the length of the ``reset word'' is not optimal. 
For example in the proof of \prref{prop:staiger} we can choose  the
word $v_1$ to be a letter, because there must be a state with distance at most one to $0$. The precise bound is ${{n+1} \choose 2} = (n+1)n/2$ if the alphabet is allowed to grow with $n$ \cite[Theorem 6.1]{Rystsov97}. If the alphabet is fixed, then the lower bound for the length of $w$ is still in $n^2/4+\Omega(n)$ \cite{Martugin08}.  
\end{remark}

%%%%%%%%%%%%%%%%%%%%%%%%%%%%%%%%%%%%%
\section{Monitorable deterministic languages}\label{sec:mdl}
The class of \moni languages form a Boolean algebra and every 
$\oo$-regular set $L$ can be written as a finite union 
$L= \bigcup_{i=1}^{n}L_i\sm K_i$ where the $L_i$ and $K_i$ are deterministic $\oo$-regular,  \cite{tho90handbook}.
Thus, if $L$ is not \moni, then one of the deterministic $L_i$ or $K_i$
is not \moni. This motivates to study \moni deterministic languages more closely.

\begin{definition}\label{def:accmoni}
Let $L\sse \Sig^\oo$ be deterministic $\oo$-regular. 
A \emph{deterministic B\"uchi monitor} (\emph{DBM} for short) for $L$ is a tuple
\begin{equation*}%\label{eq:initmoni}
\cB = (Q,\Sig,\del,q_0,F,\bot,\top)
\end{equation*}
where  $\cA= (Q,\Sig,\del,q_0,F)$ is a DBA with $L = L(\cA)$ and where $(Q,\Sig,\del,q_0,\bot,\top)$ is a monitor  in the sense of \prref{eq:initmoni}
for $L$.
\end{definition}
The next proposition justifies the definition. 
\begin{proposition}\label{prop:DBM}
Let $L\sse \Sig^\oo$ be any subset. Then $L$ is a \moni deterministic $\oo$-regular language \IFF there exists a DBM for $L$. 
\end{proposition}

\begin{proof}
The direction from right to left is trivial. Thus, let $L$ be \moni and let 
$L = L(\cA)$ for some DBA
$\cA= (Q,\Sig,\del,q_0,F)$ where all states are reachable from the initial state $q_0$.
 For a state $p\in Q$ let $L(p) = L(Q,\Sig,\del,p,F)$. 
 If $L(p) = \es$, then $L(p\cdot a) = \es$; and if $L(p) = \Sig^\oo$, then $L(p\cdot a) = \Sig^\oo$. Thus, we can merge all states $p$ with 
 $L(p) = \es$ into a single non-final state $\bot$; and we can merge all all states $p$ with 
 $L(p) = \Sig^\oo$ into a single final state $\top$ without changing the accepted language. All states are of the form $q_{0} \cdot x$ for some 
 $x\in \Sig^*$; and, 
 since $L$ is \moni, for each $x$ either there is  some $y$ with 
 $xy\Sig^\oo \cap L = \es$ or  there is  some $y$ with 
 $xy\Sig^\oo \sse L$ (or both). In the former case we have 
 $q_{0} \cdot xy = \bot$ and in the latter case we have 
 $q_{0} \cdot xy = \top$. \qed
 \end{proof}

\begin{corollary}\label{cor:states}
Let $L\sse \Sig^\oo$ be a \moni deterministic $\oo$-regular language 
and $\cA$ be a DBA with $n$ states accepting $L$.
Let $\cB$ be a DBM for $L$ with state set $Q_\cB$ where 
the size of $Q_\cB$ is as small as possible. Let further 
$Q_\cR$ (resp.~$Q_\cM$) be the state set of the congruential (resp.~smallest standard)
monitor for $L$. 
Then we have 
$$n \geq \abs{Q_\cB} \geq \abs{Q_\cR} \geq\abs{Q_\cM}.$$
\end{corollary}
%%%%%%%%%%%%%
%

%\begin{proof}
%
%\end{proof}
%

\begin{example}\label{ex:anb}
Let $\Sig = \os{a,b}$ and $\Gam = \os{a,b,c,d}$.
\begin{enumerate}
\item For $n \in \N$ consider 
$L = a^nb \Sig^\oo \sm \Sig^*bb\Sig^\oo$. It is a safety property. Hence, we have 
$\ov L = L$. Moreover, $\Sig^*bb\Sig^\oo$
is a liveness property (i.e., dense). Hence $\ov {L^{\compl}} = \Sig^\oo$.
It follows that the standard monitor is just the minimal DBA for 
$L$ augmented by the state $\bot$. There are exactly $n+4$ right-congruence
classes defined by prefixes of the words $a^nba$ and  $a^nb^2$. 
We have $L(a^nb^2) = \es$. Hence reading $a^nb^2$ leads to the state $\bot$. 
This, shows that the inequalities in \prref{cor:states} become equalities 
in that example. On the other hand $b^2$ is a forbidden factor for $L$. Hence there is a $3$ state monitor for $L$. Still there is no \epi from 
the standard monitor onto that monitor, since in the standard monitor we have
 $L(a^{n+1}) = \es$ but in the $3$-state monitor $\bot$ has not an incoming  arc labeled by 
 $a$.
\item Every monitor for the language $\Sig^*(bab\cup b^3)\Sig^\oo$ has at least  $4$ states. There are three monitors with $4$ states which are pairwise non-isomorphic. 
\item Let $L = (b^*a)^\oo \cup \os{a,b}^* c \os{a,b,c}^\oo \sse \Gam^\oo$.
Then $L$ is \moni and deterministic, but not codeterministic. 
Its minimal DBM has $4$ states, but the congruential monitor 
${Q_\cR}$ has $3$ states, only. 
We have $\ov{L} = \os{a,b,c}^\oo$ and $\ov{L^{\compl}} = \Gam^\oo$.
Hence, the smallest standard monitor has two states. 
In particular, we have  $\abs{Q_\cB} >\abs{Q_\cR} >\abs{Q_\cM}$, see also 
\prref{fig:monis}.
\end{enumerate}
\end{example}
\begin{figure}[h]
  \centering
 \begin{tikzpicture}[node distance=3cm,auto,>= triangle 45,bend angle
   =45,initial text=]
  \node[state,initial] (0) {$0$};
\node (B) [left of =0] {$\cal B:$};
\node[state,accepting] (1) [right of =0] {$1$};
\node[state,accepting] (2) [right of =1] {$2$};
\node[state] (3) [below of =1] {$\bot$};
\path[->] (1) edge [bend left] node [above] {$b$} (0)
(0) edge node [above] {$a$} (1)
(0) edge [loop above] node {$b$} ()
(1) edge [loop above] node [left] {$a$} ()
(0) edge [bend angle =45,bend left] node [above left] {$c$ \phantom{qqq}} (2)
(1) edge  node [above] {$c$} (2)
(2) edge [loop above] node {$a,b,c$} ()
(0) edge [bend angle =20, bend right] node [above] {$d$} (3)
(1) edge node [above right] {$d$} (3)
(2) edge  [bend angle =20, bend left] node  {$d$} (3)
(3) edge [loop below] node {$\Gamma$} ();
\end{tikzpicture} 
\qquad \begin{tikzpicture}[node distance=3cm,auto,>= triangle 45,bend angle =
  30, initial text=]
  \node[state,initial] (0) {$0,1$};
\node (R) [left of =0] {$\cal R:$};
\node[state] (2) [right of =0] {$2$};
\node[state] (3) [below of =1] {$\bot$};
\path[->] 
(0) edge node [above] {$c$} (2)
(0) edge [loop above] node {$a,b$} ()
(2) edge [loop above] node {$a,b,c$} ()
(0) edge node [above] {$d$} (3)
(2) edge node  {$d$} (3)
(3) edge [loop below] node {$\Gamma$} ();
\end{tikzpicture}
\hspace{1cm}\begin{tikzpicture}[node distance=3cm,auto,>= triangle 45,bend angle =
  30, initial text=]
  \node[state,initial] (0) {$0,1,2$};
\node (R) [left of =0] {$\cal M:$};
\node[state] (3) [below of =0] {$\bot$};
\path[->] 
(0) edge [loop above] node {$a,b,c$} ()
(0) edge node [left] {$d$} (3)
(3) edge [loop below] node {$\Gamma$} ();
\end{tikzpicture}
  \caption{Monitors ${\cal B}$, ${\cal R}$, ${\cal M}$ for $L=L({\cal
      B})$.}
  \label{fig:monis}
\end{figure}
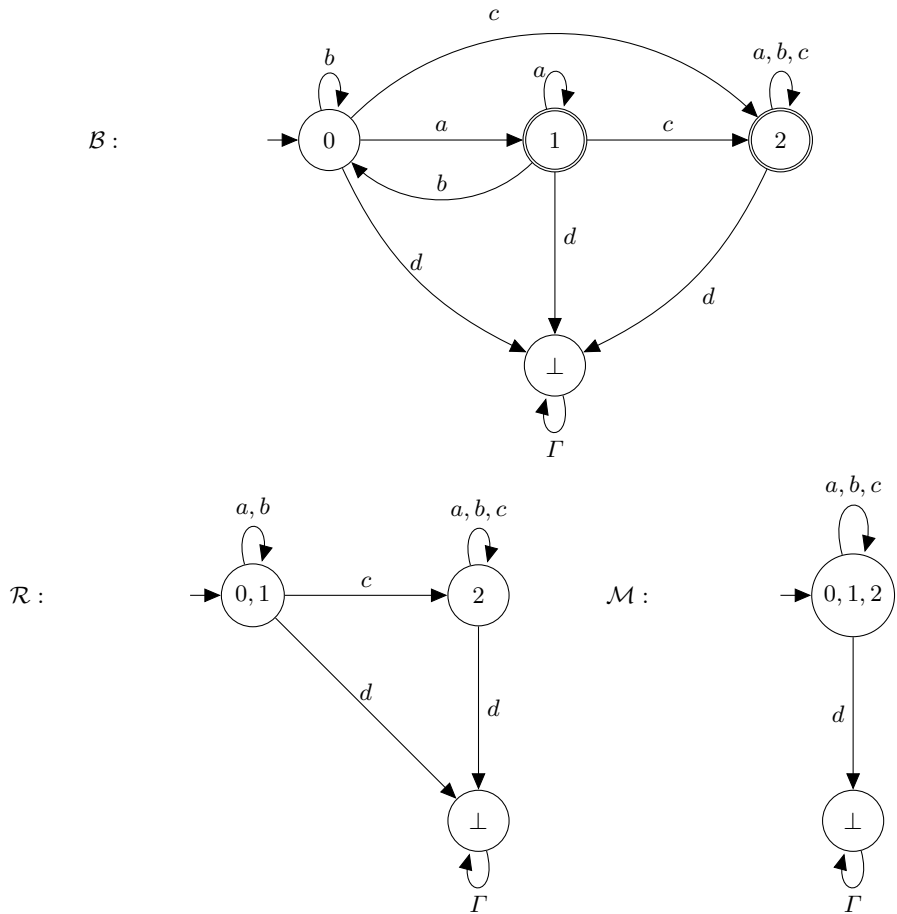

\section{Deciding liveness and monitorability}\label{sec:dec}
\subsection{Decidability for B\"uchi automata}\label{sec:Bdlm}
It is well-known that decidability of liveness (monitorability resp.) is PSPACE-complete for B\"uchi automata. The following result for liveness is classic, for monitorability
it was shown in \cite{DiekertM12_WoLLIC}.

\begin{proposition}\label{prop:lowb}
The following two problems are $\PSPACE$-complete:
\begin{itemize}
\item {\bf Input:} A B\"uchi automaton $\cA = (Q, \Sigma, \del, I, F)$.
 \item {\bf Question 1:} Is the accepted language $L(\cA) \sse \Sig^\oo$ live?
 \item {\bf Question 2:} Is the accepted language $L(\cA) \sse \Sig^\oo$ \moni?
\end{itemize}
\end{proposition}

\begin{proof}
Both problems can be checked in PSPACE using standard techniques. 
We sketch this part for monitorability.
The procedure considers, one after another, all subsets 
$P$ such that $P$ is reachable from $I$ by reading some input word. 
For each such $P$ the procedure guesses some $P'$ which is reachable from $P$.
It checks that either $L(\cA')=\es$ or $L(\cA')=\Sig^\oo$, where
$\cA'=(Q, \Sigma, \del,P', F)$. If both tests fail then the procedure 
enters a rejecting loop. 

If, on the other hand, the procedure terminates after having visited all 
$P$, then $L(\cA)$ is \moni. 

For convenience of the reader we  show \PSPACE-hardness of both problems by adapting the proof in \cite{DiekertM12_WoLLIC}.

We reduce the universality problem for non-deterministic finite
  automata (NFA) to both problems. The universality problem for NFA
  is well-known to be \PSPACE-complete.
 
Start with an NFA
$\Aa = (Q', \Gam, \del', q_0, F')$ where $\Gam \neq \es$. We use a new letter $b\notin \Gam$ and we let $\Sig = \Gam \cup \os{b}$.

We will construct B\"uchi automata  $\Bb_1$ and $\Bb_2$ as follows. 
We use three new states $d,e,f$ and we let $Q = Q'\cup \os{d,e,f}$, see 
\prref{fig:tikzfBA}.
%The repeated (or final) states of $B$ are defined as 
%$F= \os{e,f}$. 
The initial state is the same as before: $q_0$. 
Next, we define $\del$. We keep all arcs from $\del'$ and we add the 
following new arcs. 
\begin{itemize}
\item $ q \act{b} d \act{a} e \act{a} e$ for all $q \in Q' \sm F'$ and
  all $a \in \Gam$.  
 \item $ e \act{b} d \act{b} d$
 \item $ q \act{b} f \act{c} f $ for all $q \in F'$ and 
 all $c \in \Sig$. 
 \end{itemize}

%
%\begin{figure}[h]
%  \centering
%\begin{tikzpicture}[initial text=, auto, shorten >=1pt, >=latex]
%  % Split ellipse 3cm x 1cm
%  \path[draw] (0, 0) arc [start angle=40, end angle=230, x radius=3cm, y radius=1cm] coordinate[pos=0.9] (non-final-anchor) -- cycle node[pos=0.5] 
%  {$Q' \setminus F'$};
%  \path[draw, double] (0.25cm, -0.125cm) arc [start angle=40, end angle=-130, x radius=3cm, y radius=1cm] coordinate[pos=0.4] (final-anchor) -- cycle node[pos=0.5, 
%  below right] {$F'$};
%
%  % Other nodes
%  \node[state, below left=of non-final-anchor] (d) {$d$};
%  \node[state, right=of d] (e) {$e$};
%  \node[state, below right=of final-anchor] (f) {$f$};
%  
%  \path[->] (non-final-anchor) edge node[above] {$b$} (d)
%            (d) edge[loop left] node {$b$} (d)
%            (d) edge node {$\Gamma$} (e)
%            (e) edge[bend left] node {$b$} (d)
%            (e) edge[loop right] node {$\Gamma$} (e)
%            ;
%  
%  \path[->] (final-anchor) edge node {$b$} (f)
%            (f) edge[loop right] node {$\Sigma$} (f)
%            ;
%\end{tikzpicture}
%  \caption{PSPACE-hardness for liveness and monitorability for B\"uchi
%    automata.}
%\label{fig:tikzfBA}
%\end{figure}
%

\begin{figure}[h]
  \centering
  \includegraphics[scale=1.1]{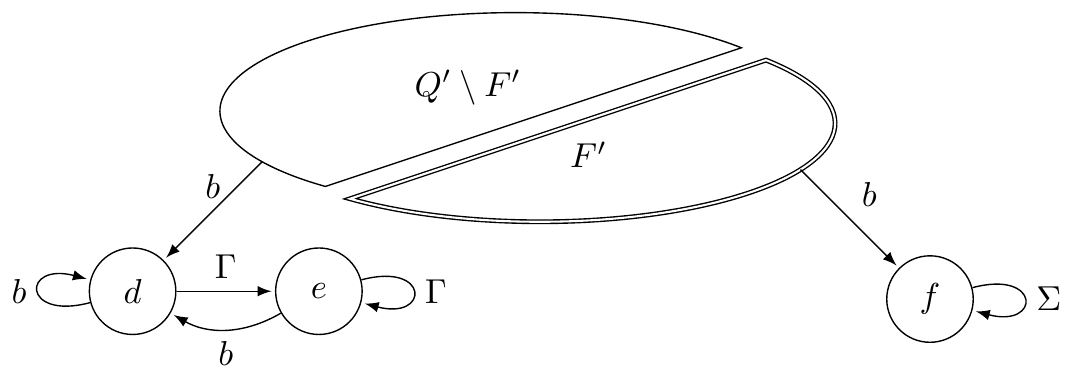}
  \caption{PSPACE-hardness for liveness and monitorability for B\"uchi
    automata.}
\label{f-BA}
\end{figure}
Let us  define two final sets of states: $F_1 = \os f$ and $F_2= \os{d,f}$. 
Thus, we have constructed  B\"uchi automata  $\Bb_1$ and $\Bb_2$ 
where 
\[\cB_i = (Q, \Gam, \del, q_0, F_i)\, \text{ for } \,i = 1,2.\]
For the  proof of the proposition it is enough to verify the 
following two claims which are actually more precise than needed. 
\begin{enumerate}
\item The language $L(\cB_1)$ is \moni. It is live \IFF $L(\cA) = \Gam^*$.
\item The language $L(\cB_2)$ is live. It is \moni \IFF $L(\cA) = \Gam^*$.
\end{enumerate}
% In order to understand the construction, consider any finite word 
% $w\in \Sig^*$. It can be written as $uu'$ where $u\in
%  \Gam^*$ is the maximal prefix without any occurrence of $b$. Possibly 
%  $u=w$. 
 
 If $L(\cA) = \Gam^*$, then we have $L(\cB_1)=L(\cB_2)=\Sig^\oo$, so
% $$w b^\oo \in wb\Sig^\oo \sse  L(\cB_1) \sse  L(\cB_2).$$ 
 both languages are live and \moni. 
 
If $L(\cA) \neq  \Gam^*$, then there exists some word $u\notin L(\cA)$ and
hence reading $ub$ we are necessarily in state $d$. It follows that 
$ub\Sig^\oo \cap L(\cB_1) = \es$ and $L(\cB_1)$ is not live.
Still, $L(\cB_1)$ is \moni.
Now, for all $w \in \Sig^*$ we have $w b^\oo \in L(\cB_2)$. Hence, $L(\cB_2)$ is live.
However, if $u\notin L(\cA)$, then after reading $ub$ we are in state $d$.
Now, choose some letter $a \in \Gam$. 
For all $v\in \Sig^*$ we have 
$ubva^\oo \notin L(\cB_2)$, but $ubvb^\oo \in L(\cB_2)$. Hence, 
if $L(\cA) \neq  \Gam^*$, then $L(\cB_2)$ is not \moni. \qed
\end{proof}

\subsection{Decidability for \LTL}\label{sec:Bdlm}
We use the standard syntax and semantics of the linear temporal logic $\LTL$ 
for infinite words over some finite nonempty alphabet $\Sig$. We restrict ourselves the pure future fragment and 
the syntax of   $\LTL_\Sigma[\XU]$ is given as follows. 
$$
\varphi ::= \True \mid a \mid \neg\varphi \mid \varphi\vee\varphi 
              \mid \varphi\XU\varphi,
$$
where $a$ ranges over $\Sigma$. 
The binary operator $\XU$ is called
the \emph{next-until} modality.

In order to give the semantics we identify each 
$\phi \in \LTLS$ with some first-order formula $\phi(x)$ in at most 
one free variable. 
The identification is done as usual by structural induction. 
The formula $a$ becomes $a(x)= P_a(x)$, where $P_a(x)$ is the unary predicate 
saying that the label of position $x$ is the letter $a$. 
The formula ``$\phi$ \emph{neXt-Until} $\psi$'' is defined by:
\begin{align*}
  (\varphi\XU\psi)(x) &= \exists z:\; (x < z \wedge 
  \psi(z) \wedge \forall y:\;  \phi(y) \vee y \leq x \vee z \leq y).
\end{align*}

Finally let $\alp \in \Sig^\oo$ be an infinite word with the first position 
$0$, then we define $\alp \models \phi$ by $\alp \models \phi(0)$; and we define
$$L(\phi)= \set{\alp \in \Sig^\oo}{\alp \models \phi}.$$
Languages of type $L(\phi)$ are called \emph{$\LTL$ definable},
It is clear that every $\LTL$ definable language is first-order definable; and 
Kamp's famous theorem \cite{kam68} states the converse. In particular, 
given $L(\phi)$ there exists a BA $\cA$ such that $L(\phi)= L(\cA)$.
There are examples where the size of the formula $\phi$ is exponentially smaller than the size of any corresponding BA $\cA$.

For a survey on first-order definable languages we refer to 
\cite{dg08SIWT}.
By $\LTL$ decidability of a property $\cP$ we mean that the input is a 
formula $\phi\in \LTLS$ and we ask whether property $\cP$ holds for 
$L(\phi)$. By \prref{prop:lowb} we obtain straightforwardly the following lower and upper bounds for the $\LTL$ decidability of monitorability 
and liveness. 
\begin{remark}\label{rem:lowup}
The following two problems are $\PSPACE$-hard and can be solved in EXPSPACE:
\begin{itemize}
\item {\bf Input:} A formula $\phi\in \LTLS$.
 \item {\bf Question 1:} Is the accepted language $L(\phi) \sse \Sig^\oo$ live?
 \item {\bf Question 2:} Is the accepted language $L(\phi) \sse \Sig^\oo$ \moni?
\end{itemize}
\end{remark}
\prref{rem:lowup} is far from satisfactory since there is huge gap between $\PSPACE$-hardness and containment in EXPSPACE. 
Very unfortunately, we were not able to make the gap any smaller for 
 monitorability. There was some belief in the literature that, at least, $\LTL$ liveness can be tested in $\PSPACE$, see 
for example \cite{NitscheW97}. But surprisingly this last assertion is wrong: 
 testing $\LTL$ liveness is 
EXPSPACE-complete! 
\begin{proposition}\label{prop:livees}
Deciding $\LTL$ liveness is EXPSPACE-complete:
\begin{itemize}
\item {\bf Input:} A formula $\phi\in \LTLS$.
\item {\bf Question} Is the accepted language $L(\phi) \sse \Sig^\oo$ live?
\end{itemize}
\end{proposition}
EXPSPACE-completeness of liveness was proved by Muscholl and
Walukiewicz in 2012, but never published. Independently, it was
 proved by Orna Kupferman and 
Gal Vardi in \cite{kv15CSL}.
 
We give a proof of \prref{prop:livees} in Sections~ \ref{sec:eexps} and~\ref{sec:propproof}
below. We also point out why the proof technique fails to say anything about the hardness to decide monitorability.
Our proof for \prref{prop:livees}
 is generic. This means
that we start with a Turing machine $M$
which accepts a language $L(M)\sse \Gam^*$ in EXPSPACE. We show that 
we can construct in polynomial time a formula $\phi(w)\in \LTLS$ such that
\[w \in L(M) \iff L(\phi(w)) \sse \Sig^\oo \text{ is not live}.\]

\subsection{Encoding EXPSPACE computations}\label{sec:eexps}
For the definition of Turing machines we use standard conventions, very closely to the notation e.g.~in \cite{HU}. Let
$L=L(M)$ be accepted by a deterministic Turing machine $M$, % = (Q,\Sig,\Gam,\del,q_0,\os{q_f},B)$
where $M$
has set of states $Q$ and the tape alphabet is $\Gam$ containing a ``blank'' symbol $B$.
We assume that for some fixed polynomial $p(n)\geq n+2$ the machine $M$ uses on an input word $w\in (\Gam\sm \os B)^*$
of length $n$ strictly less space than $2^N-2$, where $N=p(n)$. (It does not really matter that $M$ is deterministic.)  Configurations are words from
$\Gam^* (Q \times \Gamma) \Gam^*$ of length precisely $2^{N}$, where
the head position corresponds to the symbol from $Q \times
\Gamma$. For technical reasons  we will assume that the first and the
last symbol in each configuration is $B$.  Let $A=\Gamma \cup (Q \times \Gamma)$.

 If the input is nonempty word $w=a_1\cdots a_{n}$ where the $a_i$ are letters, then 
 the initial configuration is defined here as
\[
C_0 = B(q_0,a_1)a_2 \cdots a_n\underbrace{BBBBB \cdots B}_{2^{N}-n-1
  \text{ times}}.
\]

For $t\geq 0$ let $C_t$ be configuration of $M$ at time $t$ during the
computation starting with the initial configuration $C_0$ on input $w$. We may
assume that the computation is successful \IFF there is some $t$ such
that a special symbol, say $q_f$, appears in $C_t$.  Thus, we can write each $C_t$ as a word
$C_t = a_{0,t} \cdots a_{m,t}$ with $m = 2^N-1$; and we have $w \in
L(M)$ \IFF there are some $i\geq 1$ and $t\geq 1$ such that   
$a_{i,t} = q_f$.

In order to check that a sequence $C_0$, $C_1$, \ldots  is a valid
computation we may assume that the Turing machine comes with a
table $\Del \sse A^4$ such that the following formula holds:
\[\forall t>0 \; \forall 1 \le i < 2^N-1:
(a_{i-1,t-1},a_{i,t-1},a_{i+1,t-1},a_{i,t}) \in \Del.
\]
Without restriction we have $(B,B,B,B) \in \Delta$, because otherwise 
$M$ would accept only finitely many words. 

 We may express that we can reach a final  configuration $C_t$  by saying: 
 \[ \exists t\geq 1\; \exists 1 \leq i < 2^N: a_{i,t} = q_f.\] As
 in many EXPSPACE-hardness proofs, for comparing successive
 configurations we need to switch to a slightly different encoding, by
 adding the tape position after each symbol from $A$. To do so, we
 enlarge the alphabet $A$ by new symbols $0,1,\$,\#,k_1, \ldots k_N$
 which are not used in any $C_t$ so far. Hence, $\Sig = A \cup
 \os{0,1,\$,\#,k_1, \ldots k_N}$. We encode a position $0 \leq i <
 2^N$ by using its binary representation with exactly $N$ bits. Thus,
 each $i$ is written as a word $\bin(i) = b_1\cdots b_N$ where each
 $b_p\in \{0,1\}$.  In particular, $\bin(0) = 0^N$, $\bin(1) =
 0^{N-1}1$, \ldots, $\bin(2^N-1) = 1^N$.

Henceforth, a configuration $C_t= a_{0,t} \cdots a_{m,t}$ with $m= 2^N-1$ is encoded as a word
\[c_t= a_{0,t}\bin(0) \cdots a_{m,t}\bin(m)\$.\]
Words of this form are called \emph{stamps} in the following.
Each stamp has length $2^N\cdot N +1$. If a factor 
$\bin(i)$ occurs, then either $i=m$ (i.e., $\bin(i) = 1^N$) and the next letter is $\$$ or 
$i<m$ and the next letter is some letter from the original alphabet
$A$ followed by the word $\bin(i+1)$.

Now we are ready to define a language $L= L(w)$ which has the property that 
$L$ is not live \IFF $w \in L(M)$. 
We describe the words $\alp \in \Sigma^\oo$ which belong to $L$ as follows.  
\begin{enumerate}
\item Assume that $\alp$ does not  start with a prefix of the form $c_0\cdots
  c_\ell\#$, where $c_0$ corresponds to the initial configuration
  w.r.t.~$w$, each $c_t$ is a stamp and in the stamp $c_\ell$ the
  symbol $q_f$ occurs. Then $\alp$ belongs to $L$.

\item Assume now that $\alp$ starts with a prefix $c_0\cdots
  c_\ell\#$ as above.
Then we let $\alp$ belong to $L$ if and only if the set of letters
occurring infinitely often in $\alpha$ witness that the prefix $c_0\cdots c_\ell$
of stamps is {\bf not} a valid computation.  Thus, we must point to some $t\geq 1$ and 
some position $1\leq i< m$ such that 
$(a_{i-1,t-1},a_{i,t-1},a_{i+1,t-1},a_{i,t}) \notin \Del.$
The position $i$ is given as $\bin(i) = b_1\cdots b_N \in \os{0,1}^N$.
The string $\bin(i)$ 
defines a subset of $\Sigma$:
\[
I(i)= \set{k_p \in \os{k_1,\ldots, k_N}}{b_p= 1}.
\]
The condition for $\alp$ to be in  $L$ is that 
for some $t$ the mistake from $c_{t-1}$ to $c_t$ is reported by
$(a_{i-1,t-1},a_{i,t-1},a_{i+1,t-1},a_{i,t}) \notin \Del$ and
the position  $i$ such that $I(i)$ equals the set of
letters $k_p$ which appear infinitely often in $\alp$. Note that since
we excluded  mistakes at positition $i=0$ (because of the leftmost $B$),
the set $I(i)$ is non-empty.
\end{enumerate}

\begin{lemma}\label{lem:live}
The language $L= L(w)$ is not live \IFF $w \in L(M)$. 
\end{lemma}

\begin{proof}
%We have to show two directions. 

First, let $w\in L(M)$. Then we claim that
$L$ is not live. To see this let  
$u= c_0\cdots c_\ell\#$, where the prefix 
$c_0\cdots c_\ell$ is a valid accepting computation of $M$. There is no mistake 
in $c_0\cdots c_\ell$. Thus we have $u \Sig^\oo \cap L= \es$, so indeed, $L$ is not live. 

Second, let $w\notin L(M)$. We claim that $L$ is live. 
Consider any $u \in \Sigma^*$. Assume first that  $u$ does not start with a prefix of
the form  $c_0\cdots
  c_\ell\#$, where $c_0$ corresponds to the initial configuration
  w.r.t.~$w$, each $c_t$ is a stamp and in the stamp $c_\ell$ the
  symbol $q_f$ occurs. Then we %can find some word $v \in\Sigma^*$ such that 
  we have $u\Sigma^\oo \subseteq L$. 

Otherwise, assume that $c_0\cdots
  c_\ell\#$  is a prefix of $u$ and that all $c_t$'s are stamps, with $c_0$
  initial and $c_\ell$ containing  $q_f$. 
% Consider any $\alp \in \Sig^\oo$. If $\alp$ does not have a prefix 
% $c_0\cdot c_\ell\#$ then we can choose a prefix $\gam$ of length $2^n\cdot n +1$ and we have $\gam \Sig^\oo \sse L$. Thus, we may assume that 
% $c_0\cdot c_\ell\#$  is a prefix and all $c_t$'s are stamps.
  There must be some mistake in $c_0\cdots
  c_\ell\#$, say for some $i$ and $t$. Let $I(i)$ be as defined
  a above.  As $i \geq 1$ we have $I(i) \neq \es$. Therefore we
  let $\bet$ be any infinite word where the set of letters appearing
  infinitely often is exactly the set $I(i)$. By definition
  of $L$ we have $u \bet \in
    L$. Hence, $L$ is live.  \qed
\end{proof}

There are other ways to encode EXPSPACE computations which may serve to prove 
\prref{prop:livees}, see for example \cite{kv15CSL}. However, these proofs do not reveal any hardness for $\LTL$ monitorability. In particular, they do not reveal 
EXPSPACE or EXPTIME hardness. 
For our encoding this can made very precise. 

\begin{remark}\label{rem:nomon}
Since are interested in  EXPSPACE-hardness, we may assume  that there infinitely many $w$ with $w \notin L(M)$.
Let $n$ be large enough, say $n\geq 3$ and $w \notin L(M)$, then 
$(B,(q_0,a_1),a_2,q_f) \notin \Del$, where $w = a_1a_2\cdots $ because otherwise $w \in L(M)$.
Define $c_1$ just as the initial stamp $c_0$ with the only difference that the letter $(q_0,a_1)$ is replaced by the symbol
$q_f$. Let $u = c_0 c_1 \#$, then 
  for every $v \in \Sigma^*$ we have that $uv(k_N)^\oo \in L$ 
(i.e.,
  there is a mistake at position 1), but 
   $uv(k_1 k_2 \cdots k_{N})^\oo \cap L = \es$
  (i.e., there is no mistake at position $2^N-1$) because  $(B,B,B,B) \in \Del$.   Thus, $L$ is not monitorable.
  \end{remark}

\subsection{Proof of Proposition \ref{prop:livees}}\label{sec:propproof}
$\LTL$ liveness is in EXPSPACE by \prref{rem:lowup}. 
The main ideas for the proof are in the previous subsection. We  show that we can construct in polynomial time on input $w$ some  $\phi \in \LTLS$ such that $L(\phi)= L(w)$. This can be viewed as a standard exercise in $\LTL$.
The solution is a little bit tedious and leads to a formula of at most quadratic size in $n$.
The final step in the proof is to apply \prref{lem:live}. \qed

\section{Conclusion and outlook}\label{sec:out}
In the paper we studied  \moni languages from the perspective of what is a 
``good monitor''. In some sense we showed that there is no final answer yet, but monitorability is a field where various interesting questions remain to be answered. 

Given an $\LTL$ formula for a monitorable property one can construct 
monitors of at most doubly exponential size; and there is some indication that this is the best we can hope for, see \cite{BauerLS06b}. 
Still, we were not able to prove any hardness for $\LTL$ monitorability 
beyond $\PSPACE$. This does not mean anything, but at least in theory, it could be that $\LTL$ monitorability cannot be tested in EXPTIME, but nevertheless it is not EXPTIME-hard. 
%
%Recall that sparse languages it might be undecidable, but still a sparse language is never $\NP$-hard (unless P = NP), see \cite{Mahaney82}. So, the difficulty of a problem does not mean that hardness results follow. 

There is also another possibility. 
Deciding monitorability might be easier than constructing a monitor. Remember that deciding monitorability means to test that the boundary is nowhere dense. However we have argued that a DBA for the boundary does not give necessarily any information about a possible monitor, see the discussion at the beginning of \prref{sec:gcm}. 

A more fundamental question is about the notion of monitorability.
The definition is not robust in the sense that every language becomes monitorable simply by embedding the language into a larger alphabet. 
This is somewhat puzzling, so the question is whether a more robust and still useful notion of monitorability exist. 

Finally, there is an interesting connection to learning. In spite of recent 
progress to learn general $\oo$-regular languages by \cite{AngluinF14} it not known how to learn a DBA for 
deterministic $\oo$-regular languages in polynomial time. The best result is still due to Maler and Pnueli in \cite{MalerP95}. They show that it is possible to learn a DWA for a $\oo$-regular language $L$
in $G_\del \cap F_\sig$ in polynomial time. The queries to the oracle are membership 
question ``$uv^\oo \in L$?'' where $u$ and $v$ are finite words and the query whether a proposed 
DWA is correct. If not, the oracle provides a shortest counterexample of the form  $uv^\oo$. 

Since a DWA serves also as a monitor we can learn a monitor the very same way, but beyond 
$G_\del \cap F_\sig$  it is not known that  membership queries to $L$ and queries whether a proposed 
monitor is correct suffice. As a first step one might try find out how
to learn a deterministic B\"uchi 
monitor in case it exists. This is a natural class beyond $G_\del \cap F_\sig$ because canonical minimal DBA for these languages exist. Moreover, just as for DWA this minimal DBA is an DBM, too. 

Another interesting branch of research is monitorability in a distributed setting. 
A step in this direction for infinite Mazurkiewicz traces was outlined in \cite{DiekertM12_WoLLIC}.

\subsection*{Acknowledgment}
The work was done while the first author was visiting LaBRI at the  Universit{\'e} Bordeaux in June 2015. The hospitality of LaBRI and their members is greatly acknowledged. 

The authors thank Andreas Bauer who communicated to us (in June 2012) that the complexity of
 $\LTL$-liveness should be regarded as open because published proofs 
 stating PSPACE-completeness were not convincing. We also thank Ludwig Staiger, Gal Vardi, and Mikhail Volkov for helpful comments. 
 
%
%%
%%%%%%%%%%%%%%%%%%%%%%%
%\bibliographystyle{abbrv}
%\bibliography{../monitoring,../../traces,../../lit}
%%\bibliography{../../traces}

\begin{thebibliography}{10}

\bibitem{AngluinF14}
D.~Angluin and D.~Fisman.
\newblock Learning regular omega languages.
\newblock In P.~Auer, A.~Clark, T.~Zeugmann, and S.~Zilles, editors, {\em
  Algorithmic Learning Theory - 25th International Conference, {ALT} 2014,
  Bled, Slovenia, October 8-10, 2014. Proceedings}, volume 8776 of {\em Lecture
  Notes in Computer Science}, pages 125--139. Springer, 2014.

\bibitem{BauerLS06b}
A.~Bauer, M.~Leucker, and C.~Schallhart.
\newblock Monitoring of real-time properties.
\newblock In {\em Proceedings of FSTTCS'06}, number 433 in LNCS, pages
  260--272. Springer, 2006.

\bibitem{dg08SIWT}
V.~Diekert and P.~Gastin.
\newblock First-order definable languages.
\newblock In J.~Flum, E.~Gr{\"a}del, and {\Th}.~Wilke, editors, {\em Logic and
  Automata: History and Perspectives}, Texts in Logic and Games, pages
  261--306. Amsterdam University Press, 2008.

\bibitem{DiekertL2014tcs}
V.~Diekert and M.~Leucker.
\newblock Topology, monitorable properties and runtime verification.
\newblock {\em Theoretical Computer Science}, 537:29--41, 2014.
\newblock Special Issue of ICTAC 2012.

\bibitem{DiekertM12_WoLLIC}
V.~Diekert and A.~Muscholl.
\newblock On distributed monitoring of asynchronous systems.
\newblock In C.-H.~L. Ong and R.~J. G.~B. de~Queiroz, editors, {\em WoLLIC},
  volume 7456 of {\em Lecture Notes in Computer Science}, pages 70--84.
  Springer, 2012.

\bibitem{GondiPS09}
K.~Gondi, Y.~Patel, and A.~P. Sistla.
\newblock Monitoring the full range of $\omega$-regular properties of
  stochastic systems.
\newblock In N.~D. Jones and M.~M{\"u}ller-Olm, editors, {\em VMCAI}, volume
  5403 of {\em Lecture Notes in Computer Science}, pages 105--119. Springer,
  2009.

\bibitem{HU}
J.~E. {Hopcroft} and J.~D. {Ulman}.
\newblock {\em Introduction to Automata Theory, Languages and Computation}.
\newblock Addison-Wesley, 1979.

\bibitem{kam68}
H.~Kamp.
\newblock {\em Tense Logic and the Theory of Linear Order}.
\newblock PhD thesis, University of California, 1968.

\bibitem{KMP}
D.~{Knuth}, J.~H. {Morris}, and V.~{Pratt}.
\newblock Fast pattern matching in strings.
\newblock {\em SIAM J. Comput.}, 6:323--350, 1977.

\bibitem{kv15CSL}
O.~Kupferman and G.~Vardi.
\newblock On relative and probabilistic finite counterabilty.
\newblock In {\em Proceedings 24th EACSL Annual Conference on Computer Science
  Logic (CSL 2105)}. Springer, 2015.
\newblock To appear.

\bibitem{Landweber69}
L.~H. Landweber.
\newblock Decision problems for $\omega$-automata.
\newblock {\em Mathematical Systems Theory}, 3(4):376--384, 1969.

\bibitem{LeuckerS08jlap}
M.~Leucker and C.~Schallhart.
\newblock A brief account of runtime verification.
\newblock {\em Journal of Logic and Algebraic Programming}, 78(5):293--303,
  may/june 2009.

\bibitem{MalerP95}
O.~Maler and A.~Pnueli.
\newblock On the learnability of infinitary regular sets.
\newblock {\em Inf. Comput.}, 118:316--326, 1995.

\bibitem{Martugin08}
P.~V. Martugin.
\newblock A series of slowly synchronizing automata with a zero state over a
  small alphabet.
\newblock {\em Information and Computation}, 206:1197--1203, 2008.

\bibitem{mat73}
{\Yu}.~Matiyasevich.
\newblock Real-time recognition of the inclusion relation.
\newblock {\em Journal of Soviet Mathematics}, 1:64--70, 1973.
\newblock Translated from Zapiski Nauchnykh Seminarov Leningradskogo Otdeleniya
  Matematicheskogo Instituta im.\ V.~A.~Steklova Akademii Nauk SSSR, Vol.~20,
  pp.~104--114, 1971.

\bibitem{NitscheW97}
U.~Nitsche and P.~Wolper.
\newblock Relative liveness and behavior abstraction ({E}xtended abstract).
\newblock In J.~E. Burns and H.~Attiya, editors, {\em Proceedings of the
  Sixteenth Annual {ACM} Symposium on Principles of Distributed Computing (PODS
  1997), Santa Barbara, California, USA, August 21-24, 1997}, pages 45--52.
  {ACM}, 1997.

\bibitem{PnueliZ06}
A.~Pnueli and A.~Zaks.
\newblock {PSL} model checking and run-time verification via testers.
\newblock In {\em Formal Methods}, volume 4085 of {\em LNCS}, pages 573--586.
  Springer, 2006.

\bibitem{Rystsov97}
I.~Rystsov.
\newblock Reset words for commutative and solvable automata.
\newblock {\em Theoretical Computer Science}, 172:273--279, 1997.

\bibitem{SistlaZF11}
A.~P. Sistla, M.~Zefran, and Y.~Feng.
\newblock Monitorability of stochastic dynamical systems.
\newblock In G.~Gopalakrishnan and S.~Qadeer, editors, {\em CAV}, volume 6806
  of {\em Lecture Notes in Computer Science}, pages 720--736. Springer, 2011.

\bibitem{Staiger76a}
L.~Staiger.
\newblock Regul{\"{a}}re {N}ullmengen.
\newblock {\em Elektronische Informationsverarbeitung und Kybernetik},
  12(6):307--311, 1976.

\bibitem{Staiger83}
L.~Staiger.
\newblock {Finite-state $\omega$-languages.}
\newblock {\em Journal of Computer and System Sciences}, 27:434--448, 1983.

\bibitem{sw74eik}
L.~Staiger and K.~W. Wagner.
\newblock Automatentheoretische und automatenfreie {C}harakterisierungen
  topologischer {K}lassen regul{\"a}rer {F}olgenmengen.
\newblock {\em Elektronische Informationsverarbeitung und Kybernetik},
  10:379--392, 1974.

\bibitem{tho90handbook}
W.~Thomas.
\newblock Automata on infinite objects.
\newblock In J.~van Leeuwen, editor, {\em Handbook of Theoretical Computer
  Science}, chapter~4, pages 133--191. Elsevier Science Publishers B.\ V.,
  1990.

\bibitem{Volkov08}
M.~V. Volkov.
\newblock Synchronizing automata and the {{\v C}ern{\'y}} conjecture.
\newblock In C.~Mart{\'{\i}}n{-}Vide, F.~Otto, and H.~Fernau, editors, {\em
  Language and Automata Theory and Applications, Second International
  Conference, {LATA} 2008, Tarragona, Spain, March 13-19, 2008. Revised
  Papers}, volume 5196 of {\em Lecture Notes in Computer Science}, pages
  11--27. Springer, 2008.

\bibitem{Wagner79}
K.~W. Wagner.
\newblock On omega-regular sets.
\newblock {\em Information and Control}, 43:123--177, 1979.

\end{thebibliography}
%\end{document}

\newcommand{\Ju}{Ju}\newcommand{\Ph}{Ph}\newcommand{\Th}{Th}\newcommand{\Ch}{Ch}\newcommand{\Yu}{Yu}\newcommand{\Zh}{Zh}\newcommand{\St}{St}\newcommand{\curlybraces}[1]{\{#1\}}\def\Nst#1{$#1^{st}$}\def\Nnd#1{$#1^{nd}$}\def\Nrd#1{$#1^{rd}$}\def\Nth#1{$#1^{th}$}

\end{document}